\documentclass[twocolumn,pra,superscriptaddress,amssymb,showpacs,nofootinbib]{revtex4}
\usepackage{amsmath,color,graphicx,amsthm,bbm}
\usepackage{hyperref}

\newcommand{\tmop}[1]{\ensuremath{\operatorname{#1}}}

\newcommand{\ket}[1]{\mbox{$| #1 \rangle$}}
\newcommand{\bra}[1]{\mbox{$\langle #1 |$}}
\newcommand{\proj}[1]{\mbox{$| #1 \rangle\!\langle #1 |$}}
\def\H{\mathcal{H}}

\def\s{\mathcal{S}^+}

\def\t{^{\mbox{\tiny T}}}
\def\c{\mathcal{C}^{\mbox{\tiny CHSH}}_{12}}
\def\nv{\bf{0}}
\def\tr{\mbox{tr}}

\def\o{\!\otimes\!}
\def\id{\mathbbm{1}}
\newcommand{\C}{{\mathbb{C}}}
\def\vs{\vspace{.2cm}}

\newtheorem{theorem}{Theorem}
\newtheorem{lemma}[theorem]{Lemma}
\newtheorem{definition}[theorem]{Definition}

\begin{document}

\title{All entangled states display some hidden nonlocality}

\author{Yeong-Cherng~Liang} 
\affiliation{Group of Applied Physics, University of Geneva, CH-1211 Geneva 4, Switzerland}

\author{Llu\'{\i}s~Masanes}
\affiliation{ICFO-Institut de Ci\`encies Fot\`oniques, Av. Carl Friedrich Gauss 3, E-08860 Castelldefels (Barcelona), Spain}

\author{Denis Rosset} 
\affiliation{Group of Applied Physics, University of Geneva, CH-1211 Geneva 4, Switzerland}

\date{\today}

\begin{abstract}
A well-known manifestation of quantum entanglement is that it may lead to correlations 
that are inexplicable within the framework of a locally causal theory --- a fact that is demonstrated by
the quantum violation of Bell inequalities. The precise relationship between quantum entanglement and the violation of Bell inequalities is, however, not well understood. While it is known that entanglement is necessary for such a violation, it is not clear whether all entangled states violate a Bell inequality, even in the scenario where one allows joint operations on multiple copies of the state and local filtering operations before the Bell experiment. In this paper
we show that {\em all} entangled states, namely, all non-fully-separable states of arbitrary Hilbert space dimension and arbitrary number of parties, violate a Bell inequality when combined with another state which on its own cannot violate the same Bell inequality. This result shows that quantum entanglement and quantum nonlocality are in some sense equivalent, thus giving an affirmative answer to the aforementioned open question. It follows from 
our result that two entangled states that are apparently useless in demonstrating quantum nonlocality via 
a specific Bell inequality can be combined to give a Bell violation of the same inequality. 
Explicit examples of such activation phenomenon are provided.
\end{abstract}

\maketitle


\section{Introduction}

Quantum nonlocality, i.e., the violation of Bell inequalities~\cite{Bell:1964} by entangled quantum states, is one of the most astonishing features predicted by quantum theory. Since Bell inequalities are constraints on measurement statistics that follow directly from the intuitive notion of local causality~\cite{LC}, its experimental violation under strict locality condition~\cite{Exp} suggests that an intuitive, causal explanation of certain quantum phenomena may be out of reach (see Ref.~\cite{HiddenInfluence} for an analogous conclusion even if quantum theory is not entirely correct).

Although entanglement is necessary~\cite{Werner:1989} for the demonstration of quantum nonlocality, and  all pure entangled states violate some Bell inequalities~\cite{N.Gisin:PLA:1992,S.Popescu:PLA:1992,Yu:2012} (see also Ref.~\cite{Cavalcanti:2011}),  some {\em mixed} entangled states can provably satisfy all Bell inequalities when measured one copy at a time in {\em any} Bell experiment~\cite{Werner:1989,Barrett:2002,Toth:2006,Almeida:2007}. For conciseness, we will say that these states are {\em 1-local} in the rest of this paper.  Interestingly, some of these 1-local states can become Bell-inequality-violating {\em if}, prior to the Bell experiment, an appropriate local (filtering) operation is successfully applied to the individual subsystems~\cite{SP,Gisin:1996}. This phenomenon was termed {\em hidden nonlocality}~\cite{SP} and has since been demonstrated in photonic experiments~\cite{Kwiat}.

Inspired by the idea of entanglement distillation~\cite{distillation}, a more general scheme of demonstrating hidden nonlocality was also proposed by Peres~\cite{Peres:1996} whereby the local filtering operation is allowed to act on multiple copies of identical quantum states. By this means, he showed that even some very noisy singlet state --- not known to exhibit nonlocal behavior at that time --- can indeed display hidden nonlocality. Interestingly,  it is also possible  to demonstrate the nonlocal behavior of some 1-local quantum states via joint local measurements on multiple copies of the same state without local filtering operation. This possibility was first raised as an open problem in Ref.~\cite{Liang:PRA:2006} and such examples have since been found in Refs.~\cite{Palazuelos:1205.3118, Cavalcanti:1207.5485} (see also Ref.~\cite{NV}). 

Despite all this progress, it remains unclear whether  {\em all} entangled states can exhibit non-locally-causal (henceforth abbreviated as nonlocal) correlations in the standard scenario where the experimenters can choose freely among a number of alternative measurement settings (see, however, Ref. \cite{Buscemi:PRL:2012} for a variant of this standard scenario). For example, even if we allow local filtering operations on arbitrarily many copies of the same quantum state, it was shown in Ref.~\cite{Masymp} that the set of bipartite quantum states that can violate the Clauser-Horne-Shimony-Holt (CHSH) Bell inequality~\cite{CHSH} is precisely the set of distillable~\cite{distillation} quantum states. While this does not imply that bound entangled~\cite{BoundEntangled} states must satisfy all Bell inequalities (see, eg., Ref.~\cite{Vertesi:PRL:030403}), it clearly suggests that in order to identify the nonlocal behavior of a quantum state, more general protocols are worth considering.

In Ref.~\cite{MLD}, one such possibility to demonstrate the nonlocal behavior of all bipartite entangled states was proposed: instead of local filtering operations on many copies of the same quantum state, one considers local filtering operations that acts jointly on a quantum state $\tau$  and an auxiliary state $\rho$ (which by itself does not violate some Bell inequality under consideration) prior to the Bell experiment. Within this scheme, it was shown~\cite{MLD} that for any bipartite entangled state $\tau$, there exists another state $\rho$ (which by itself does not violate the CHSH inequality even after arbitrary local filtering operations) such that $\tau\otimes\rho$ does violate the CHSH inequality after appropriate local preprocessing. Since the auxiliary state $\rho$, by construction, does not violate the CHSH inequality even when supplemented with an arbitrary amount of classical correlations, the results of Ref.~\cite{MLD} imply that for every bipartite entangled state $\tau$, there is a scenario where $\tau$ cannot be substituted by classical correlations without changing the statistics. This shows that all bipartite entangled states can exhibit some nonlocal correlations, and hence display some hidden nonlocality.

In this paper, we generalize the results of Ref.~\cite{MLD} to states involving an arbitrary number of parties, thereby showing that all non-fully-separable states are capable of exhibiting some nonlocal behavior. In Sec.~\ref{Sec:Protocol}, we describe our protocol that serves this purpose using the CHSH inequality. Then, in Sec.~\ref{Sec:Examples}, we provide some explicit examples of quantum states displaying such hidden nonlocality. Finally, we conclude with some further discussion in Sec.~\ref{Sec:Discussion}.

\section{Nonlocal behavior from all multipartite entangled states}
\label{Sec:Protocol}

Our main goal in this section is to show that all multipartite entangled states are capable of exhibiting some nonlocal behavior. To be more precise, an $n$-partite state $\rho$ is said to be entangled if it is not fully separable, i.e.,
\begin{equation}\label{fully-sep}
    \rho \neq \sum_i p_i\, \rho_1^i\otimes\cdots\otimes\rho_n^i,
\end{equation}
for any normalized, non-negative weights $p_i$ and any density matrix $\rho^i_k$ acting on the $k$-th Hilbert space $\H_k$. 

To manifest the nonlocal behavior of all multipartite entangled states, we consider  an adaptation of the standard CHSH-Bell-test to the multipartite scenario, namely, we shall allow all parties to perform stochastic local operations (SLOs) prior to the Bell test, and {\em only if} all these local operations are successful is a test of the CHSH inequality carried out between the 1st and the 2nd party (although it could be any other pair), see Figure~\ref{Fig:Scenario}. 
For the benefit of subsequent discussion, we remind that SLOs --- also known by the name of local filtering operations --- are represented, up to normalization, by separable maps~\cite{SLO}
\begin{equation}
\label{sepmap}
  \Omega(\rho) = \sum_i \left( F_1^i \otimes \cdots \otimes F_n^i \right) \rho \left( F_1^i \otimes \cdots \otimes F_n^i \right)^\dagger\ ,
\end{equation}
where ${\bf F}^i =  F_1^i \otimes \cdots \otimes F_n^i$ are the Kraus operators, and each $F_k^i$ is a matrix that acts on the $k$-th Hilbert space $\H_k$.

\begin{figure}[h!]
\scalebox{1}{\includegraphics{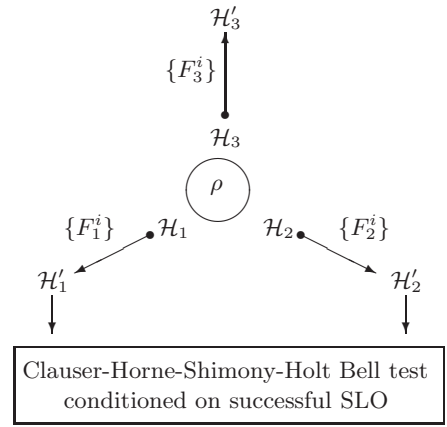}}
\caption{\label{Fig:Scenario} Schematic  illustrating the generalized CHSH Bell test in the tripartite scenario. The tripartite state $\rho$ acting on $\H_1\otimes\H_2\otimes\H_3$ first undergoes stochastic local operations (SLO) --- as described by the Kraus operators $\{F_k^i\}$ --- which can succeed with some non-zero probability. A standard CHSH Bell test is carried out on the resulting state acting on $\H_1'\otimes\H_2'$ only if the corresponding SLO is successful.}
\end{figure}

For completeness, we also remind that a CHSH-Bell test is one whereby both the experimenters have a choice over two alternative measurements (labeled by $x=0,1$ for the first party and $y=0,1$ for the second) and where each measurement outcome (labeled by $a$ for the first party and $b$ for the second) takes a binary value. In these notations, the CHSH inequality~\cite{CHSH}, which is a constraint that has to be satisfied by all locally causal correlations, reads as
\begin{equation}\label{Eq:CHSH}
	E_{00}+E_{01}+E_{10}-E_{11}\le 2, 
\end{equation}
where the correlators
\begin{equation}
	E_{xy}\equiv P(a=b|x,y)-P(a\neq b|x,y)
\end{equation}
are defined in terms of the difference between two joint conditional probabilities of the measurement outcomes. It is worth noting that other versions of the CHSH inequality can be obtained by appropriate relabeling of measurement settings and/or outcomes. Together, they form the necessary and sufficient conditions~\cite{A.Fine:PRL:1982} for the measurement statistics in this Bell scenario to admit a locally causal description.

With all the ingredients introduced above, we are now ready to define the set of $n$-partite quantum states that {\em do not} exhibit, by themselves, nonlocal behavior in the generalized CHSH-Bell test. 

\begin{definition}[CHSH-preprocessed-1-local] 
An $n$-partite state is said to be CHSH-preprocessed-1-local between parties 1 and 2 if it does not violate the CHSH-inequality between parties 1 and 2, even after $n$-partite stochastic local operations without communication. The set of $n$-partite states that are CHSH-preprocessed-1-local between party 1 and 2 is denoted by $\c$.
\end{definition}

Evidently, the set $\c$ contains all $n$-partite states that are separable between the Hilbert space of the first party $\H_1$ and the second party $\H_2$ even after arbitrary SLO and tracing out the remaining $n-2$ parties.  Moreover, it also follows from result 4 of Ref.~\cite{Masymp} that $\c$ contains all states that are bound entangled~\cite{BoundEntangled} between $\H_1$ and $\H_2$ even after arbitrary SLO and tracing out the remaining $n-2$ parties. For the $n=2$ case, the set of two-qubit states that are in $\c$ has been characterized in Ref.~\cite{VW}. For general $n$ and Hilbert space dimensions, the set of states that are CHSH-preprocessed-1-local is characterized implicitly via the following Lemma.

\begin{lemma}
\label{Lem:C}
An $n$-partite state $\rho$ acting on
$\bigotimes_{i=1}^n \H_{i}$ belongs to $\c$ if, and only if, it satisfies
\begin{equation}\label{C}
    \tr\!\left[\rho \left( \mbox{$\bigotimes_{i=1}^n$}  F_i \right) \left(H_\theta\otimes\id_{\C^2}^{n-2}\right)\,
    \left( \mbox{$\bigotimes_{i=1}^n$}  F_i \right)^\dag \right]
    \geq 0\ ,
\end{equation}
for all matrices $F_i\! :\C^2\rightarrow\H_{i}$ and all numbers $\theta \in [0,\pi/4]$, where
\begin{equation} \label{H}
    H_\theta\equiv \id_{\C^2}\otimes\id_{\C^2}-\cos\theta\,
    \sigma_x\otimes\sigma_x - \sin\theta\,
    \sigma_z\otimes \sigma_z\ ,
\end{equation}
with $\id_{\C^2}$ being the $2\times 2$ identity matrix and where 
$\{\sigma_i\}_{i=x,y,z}$ are the Pauli matrices.
\end{lemma}

The proof of this Lemma is a straightforward extension of that for Lemma 1 in Ref.~\cite{MLD} but for completeness, we have included it in Appendix~\ref{App:Lemma2}. While the characterization of $\c$ is interesting in its own, here, we are mainly interested in using it to show the central result of this paper, as summarized in the following Theorem.

\begin{theorem}\label{Thm} 
An $n$-partite state $\tau$ is entangled if and only if there exists a state $\rho \in \c$ such that $\rho\otimes\tau$ is not in $\c$.\vs
\end{theorem}

\begin{proof}[Proof of Theorem~\ref{Thm}]
If $\tau$ is {\em fully separable} then, $\rho\in\c$ implies $\rho\otimes\tau \in \c$. This is so because for any separable map $\Omega$ transforming $\rho \otimes \tau$ there is another separable map $\Omega\rq{}$ acting on $\rho$ such that $\Omega(\rho \otimes \tau) = \Omega\rq{} (\rho)$. Let us now prove the other direction of the theorem.

From now on $\tau$ is an arbitrary {\em entangled} state acting on
$\H_\tau=\bigotimes_{i=1}^n \H_{i}$. Let us show that there always exists an ancillary
state $\rho\in\c$ such that $\rho\otimes\tau \not\in \c$. Let us consider $\rho$
that acts on the $n$-partite Hilbert space
\begin{equation}\label{Eq:StateSpace}
	\H_\rho=\bigotimes_{i=1}^n \left[\H_{i}'\otimes\H_{i}''\right] 
\end{equation}
where $\H_{i}'=\H_i$, and $\H_{i}'' =\C^2$ for all $i$. 

Our aim is to prove that the state $\rho\otimes\tau$ violates
(\ref{C}) for some choice of $ F_i$ and $\theta$. In particular,
let us consider the local filtering operation described by 
\begin{gather}\label{Eq:F_i}
  F_i=\tilde{ F}_i = \ket{\Phi_{\H_{i} \H_{i}'}}\otimes \id_{\H_i''},
\end{gather}
where $\ket{\Phi_{\H_{i} \H_{i}'}} = \sum_{s=1} ^{{\rm dim} \H_{i}} ({\rm dim} \H_{i})^{-1/2} |s, s\rangle$ is the maximally-entangled state between
the spaces $\H_i$ and $\H_{i}'$ (which have the same dimension),
and $\id_{\H_i''}$ is the identity matrix acting on $\C^2$, see Figure~\ref{Fig:Protocol}. A little calculation shows that for any $\rho$,~\footnote{\label{fn:dim}In order to arrive at a simple expression in Eq.~\eqref{eq}, we have set $\H_i''=\C^2$ also for all $i=3,\ldots,n$ in Eq.~\eqref{Eq:StateSpace}. However, it should be clear from the description of the generalized CHSH Bell test (Figure~\ref{Fig:Scenario}) that it suffices to consider ancillary state acting on $\H_\rho=\bigotimes_{i=1}^2 \left[\H_{i}'\otimes\H_{i}''\right]\bigotimes_{i=3}^n \left[\H_{i}'\right]$ and local filtering operations of the form $F_i= \ket{\Phi_{\H_{i} \H_{i}'}}\otimes \id_{\H_i''}$ for $i=1,2$ and $F_i= \ket{\Phi_{\H_{i} \H_{i}'}}$ for all $i=3,\ldots,n$.  }
\begin{align}
    &\tr\!\left[\left(\rho\otimes\tau\right)\,
    \left( \mbox{$\bigotimes_{i=1}^n$} \tilde{F}_i \right) \left(H_{\pi/4}\otimes\id_{\C^2}^{n-2}\right)
    \left( \mbox{$\bigotimes_{i=1}^n$} \tilde{F}_i \right) ^\dag\right]\nonumber\\
    = &\,\nu\, \tr\!\left[\rho\, (\tau\t\otimes H_{\pi/4}\otimes\id_{\C^2}^{n-2})\right]\label{eq}
\end{align}
where $\nu$ is a positive constant, $\tau\t$ stands for the
transpose of $\tau$ and $\id_{\C^2}^{n-2}=\bigotimes_{i=1}^{n-2} \id_{\C^2}$ . From Lemma~\ref{Lem:C}, we see that $\rho\otimes\tau\not\in\c$ if for $\theta=\pi/4$, and $F_i$ defined in Eq.~\eqref{Eq:F_i}, we have
\begin{equation}\label{Eq:cond}
    \tr\!\left[\rho\, (\tau\t\otimes H_{\pi/4}\otimes\id_{\C^2}^{n-2})\right] <0
\end{equation}
For convenience, in the rest of the proof we allow $\rho$ to be
unnormalized. The only constraints on the matrices $\rho \in \c$ are:
(i) positive semi-definiteness ($\rho \in \s$), and (ii) satisfiability of all
the inequalities (\ref{C}) in Lemma~\ref{Lem:C}. $\c$ is now a convex cone, and
its dual cone is defined as
\begin{equation}\label{dual}
    \c\,^*=\{X : \tr[\rho\, X]\geq 0,\
     \forall\, \rho\in\c\}\ ,
\end{equation}
where $X$ are Hermitian matrices. Farkas' Lemma \cite{farkas} implies
that all matrices in $\c\,^*$ can be written as non-negative linear
combinations of matrices $P \in \s$ and matrices of the form
$\left( \bigotimes_{i=1}^n  F_i \right) \left(H_\theta\otimes\id_{\C^2}^{n-2}\right) \left( \bigotimes_{i=1}^n   F_i \right) ^\dag$ where $ F_i\! :\C^2\rightarrow\H_{i}'\otimes \H_{i}''$.

\begin{figure}[h!]
\scalebox{1}{\includegraphics{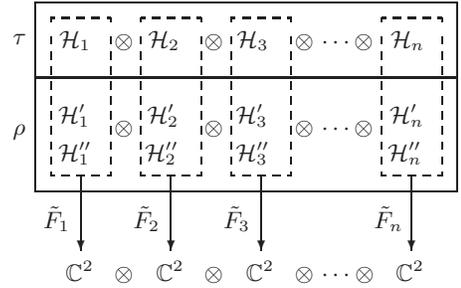}}
\caption{\label{Fig:Protocol} Schematic illustrating the local filtering operations $\tilde{ F_i}$  involved in our protocol. The solid box on top is a schematic representation of the state $\tau$ whereas that on the bottom is for the ancillary state $\rho$. Each dashed boxes enclose the subsystems possessed by the respective experimenters. For each party, the local filtering operation $\tilde{F_i}$ consists of projecting onto the maximally entangled state between the local Hilbert space $\H_i$ and $\H_i'$, while leaving the Hilbert space $\H_i''$ untouched.}
\end{figure}

We now show that there always exists $\rho\in\c$ satisfying
(\ref{Eq:cond}) by supposing otherwise and arriving at a contradiction.
Suppose that for all $\rho \in \c$ the inequality 
\begin{equation}\label{cond2}
    \tr\!\left[\rho\, (\tau\t\otimes H_{\pi/4}\otimes\id_{\C^2}^{n-2})\right] \ge0
\end{equation}
holds, and thus the matrix
$ \tau\t\otimes H_{\pi/4}\otimes\id_{\C^2}^{n-2}$ belongs to $\c\,^*$, cf. Eq.~\eqref{dual}.  Applying Farkas'
Lemma \cite{farkas} we can now write
\begin{align}\nonumber
	&\tau\t\otimes H_{\pi/4}\otimes\id_{\C^2}^{n-2}=
	\int\!\! dy\, P_y +
\\ \label{pl}
	&\int\!\! dx \left( \mbox{$\bigotimes_{i=1}^n$} F_{i,x} \right)
	\left(H_{\theta_x}\otimes\id_{\C^2}^{n-2}\right)\,
	\left( \mbox{$\bigotimes_{i=1}^n$} F_{i,x} \right)^\dag
    \ ,
\end{align}
which is equivalent to
\begin{equation}\label{principal0}
    \tau\t\! \otimes H_{\pi/4}\otimes\id_{\C^2}^{n-2} - \int\!\! dx
    \ \Omega_x\!\left( H_{\theta_x}\otimes\id_{\C^2}^{n-2} \right) \geq 0\ ,
\end{equation}
where each $\Omega_x$ is a separable map. We prove in
Lemma~\ref{Lem:SeparableMap} that inequality \eqref{principal0} requires  $\tau$ to be separable,
which gives the desired contradiction. This concludes the proof of the theorem.
\end{proof}

Theorem~\ref{Thm} is a  generalization of the theorem given in Ref.~\cite{MLD} for bipartite entangled states to any multipartite non-fully-separable state. As a result, it also inherits the dramatic consequences of the results presented therein. For instance, for any multipartite state $\tau$ that is non fully separable and which is CHSH-preprocessed-1-local, i.e., $\tau\in\c$, Theorem~\ref{Thm} implies that one can find another state $\rho\in\c$ such that $\rho\otimes\tau\not\in\c$, that is, both $\rho$ and $\tau$ display hidden-CHSH-nonlocality. This represents an example of what is now commonly referred to as the (super)activation of nonlocality~\cite{NV,Cavalcanti:2011,Palazuelos:1205.3118,Cavalcanti:1207.5485}, in the sense that by combining resources that apparently can only exhibit local behavior individually, one obtains a resource that is also capable of producing nonlocal correlations. 

There are, however, also interesting implications of Theorem~\ref{Thm} that are manifested only in the multipartite scenario. For instance, let us consider the three-qubit state $\sigma_{\mbox{\scriptsize shifts}}$ presented in Refs.~\cite{shifts,ShiftState}. This state is not fully separable (\ref{fully-sep}), but when  any two of the three parties are considered as a single one, the resulting bipartite state becomes separable, i.e., $\sigma_{\mbox{\scriptsize shifts}}$ is biseparable with respect to all bipartitions. Imagine now a 5-partite scenario where parties $\{1, 2\}$ want to violate the CHSH inequality, and parties $\{3,4,5\}$ share the supposedly useless state $\sigma_{\mbox{\scriptsize shifts} \{3,4,5\}}$. Theorem~\ref{Thm} guarantees the existence of a 5-partite state $\rho_{\{1,\ldots 5\}}$ which does not belong to $\c$, yet, together with
$\sigma_{\mbox{\scriptsize shifts} \{3,4,5\}}$ violate the CHSH-inequality between parties 1 and 2. Notice that $\sigma_{\mbox{\scriptsize shifts} \{3,4,5\}}$ does not
even involve parties $\{1, 2\}$! Therefore, in the joint ``activation" of $\rho_{\{1,\ldots 5\}}$ and $\sigma_{\mbox{\scriptsize shifts} \{3,4,5\}}$, some intricate teleportation-like~\cite{Bennett:Teleportation} phenomena between the sets of parties $\{1,2\}$ and $\{3,4,5\}$ take place.

\section{Explicit Examples}
\label{Sec:Examples}
\subsection{Bipartite examples}
\label{Sec:Bipartite}
\subsubsection{Two-qubit Werner state}
Let us now look at some explicit examples of the activation of CHSH-nonlocality in the bipartite scenario. Consider the family of 
two-qubit Werner states~\cite{Werner:1989}:
\begin{equation}
	\tau_{\mbox{\tiny W},2}(p) = p\, \proj{\Psi^-} +
	\left( 1 - p \right) \frac{\id_{\C^2}\otimes\id_{\C^2}}{4},\quad p \in \left[-\frac{1}{3}, 1 \right],
\end{equation}
which are entangled if and only if $p > \frac{1}{3}$. For the following parameter range,
\begin{equation}\label{Eq:p:interval}
	p>p^*=4 \sqrt{2} - 5 \simeq 0.6569
\end{equation}
the CHSH-nonlocality hidden in this family of states can be demonstrated via the
4-qubit ancillary state $\rho$ acting on $[\H_{1}'\otimes\H_{1}'']\otimes[\H_{2}'\otimes\H_{2}'']$
\begin{equation}\label{Eq:ActivatingState}
	\rho = \frac{1}{16} \sum_{i, j =0}^3 R_{i j} \, \sigma_i \otimes \sigma_j \otimes \sigma_i \otimes \sigma_j
\end{equation} 
where $\sigma_0 \equiv \id_2$, $\sigma_1=\sigma_x$, $\sigma_2=\sigma_y$, $\sigma_3=\sigma_z$ are the Pauli matrices, and
\begin{equation}
  \label{eqrho} R= \frac{1}{9} \left(\begin{array}{cccc}
    9 & 3 & 3 & 3\\
    1 & - 1 & 3 & - 1\\
    1 & - 1 & 3 & - 1\\
    1 & - 1 & 3 & - 1
  \end{array}\right).
\end{equation}

To see that $\rho\in\c$, it suffices to note that $\rho$ has positive partial transposition (PPT)~\cite{Peres:PPT}, thus is bound entangled~\cite{BoundEntangled} and hence  cannot even violate CHSH inequality  asymptotically (see result 4 of Ref.~\cite{Masymp}). As for $\tau_{\mbox{\tiny W},2}(p)$, since it is diagonal in the Bell basis, it follows from Ref.~\cite{Horodecki} and Theorem 3 of Ref.~\cite{VW} that $\tau_{\mbox{\tiny W},2}(p)\in\c$ for $p\le\tfrac{1}{\sqrt{2}}$. Indeed, $\tau_{\mbox{\tiny W},2}(p)$ with $p<0.6595$ cannot even violate {\em any} Bell inequality with projective measurements~\cite{Grothendieck}  .
Finally, note that for $\tau=\tau_{\mbox{\tiny W},2}(p)$ and $\rho$ defined in Eq.~\eqref{Eq:ActivatingState}, the left-hand-side of Eq.~\eqref{Eq:cond} becomes
\begin{equation}\label{Eq:Dfn:fp}
	 \tmop{tr}  \left[ \rho \left( \tau_{\mbox{\tiny W},2}(p)\t \otimes H_{\pi / 4} \right) \right]
	 = \frac{1}{12} \left( 3 - \sqrt{2} - \left( 1 + \sqrt{2} \right) p \right)
\end{equation}
which is negative for the interval of $p$ given in Eq.~\eqref{Eq:p:interval}. Hence,  $\rho\otimes\tau_p\not\in\c$, which illustrates the hidden-CHSH-nonlocality of both $\rho$ and $\tau_{\mbox{\tiny W},2}(p)$.

\subsubsection{Higher-dimensional Werner states}

More generally, a $d$-dimensional Werner state can be written as~\cite{Werner:1989}:
\begin{equation}
	\tau_{\mbox{\tiny W},d}(p) = 2p\, \frac{\Pi^-}{d(d-1)} +
	\left( 1 - p \right) \frac{\id_{\C^d}\otimes\id_{\C^d}}{d^2},
\end{equation}
where $\Pi^-=\tfrac{1}{2}(\id_{d^2}-\sum_{i,j=1}^d \ket{i}\!\bra{j}\otimes\ket{j}\!\bra{i})$ is the projector onto the antisymmetric subspace of $\C^d\otimes\C^d$ and $p$ can take any value in the interval $\left[ 1-\frac{2d}{d+1}, 1 \right]$. In Ref.~\cite{Werner:1989}, it was shown that $\tau_{\mbox{\tiny W},d}(p)$ is separable if and only if $p\le p_{\mbox{\tiny sep}}=\tfrac{1}{d+1}$ whereas for $p\le1-\tfrac{1}{d}$,  $\tau_{\mbox{\tiny W},d}(p)$ admits a local model for projective measurements.

As with the two-qubit case, it would be interesting to  identify an ancillary state $\rho\in\c$ that helps to demonstrate the nonlocal behavior of $\tau_{\mbox{\tiny W},d}(p)$. Specifically, with the local filtering protocol specified in Eq.~\eqref{Eq:cond} and if we only consider  $\rho$ that is non-negative under partial transposition, this search for $\rho\in\c$ can be formulated as a semidefinite program~\cite{SDP} (see Appendix~\ref{App:SDP} for details). Numerically, we have solved these semidefinite programs for  small values of $d$ and the critical weight $p^*$ above which such demonstration is possible is summarized in Table~\ref{Tbl:CriticalP}. Note that in all these cases, we found $p^*$ to fall in the interval of $p$ where it is known that $\tau_{\mbox{\tiny W},d}(p)$ is 1-local (under projective measurements), and where the state was not found to violate CHSH inequality even after arbitrary local filtering operation. These results suggest that the activation of $\tau_{\mbox{\tiny W},d}(p)$ and the corresponding $\rho\in\c$ may exist for even higher values of $d$, if not all $d$.

\begin{table}[h!]
	\begin{ruledtabular}
        \begin{tabular}{|c||c|c|c|c|c|}\hline
        $d $ & 2& 3 & 4 & 5 & 6\\ \hline
        $p_{\mbox{\tiny sep}}$ & $0.3333$ & 0.2500 &  $0.2000$  & 0.1667 & 0.1429\\ \hline
        $p^*$ & $0.6569$ & 0.6360 & 0.6247   & $0.6175$  &$0.6126$\\                  \hline	
        $p_{\mbox{\tiny L}}$ & 0.6595\footnote{The tighter bound of $p_{\mbox{\tiny L}}$ quoted here is due to Ref.~\cite{Grothendieck} instead of the original value of 1/2 deduced from Ref.~\cite{Werner:1989}.} & 0.6667 & 0.7500 & 0.8000 & 0.8333 \\ \hline
        $p_{\mbox{\tiny NL,SLO}}$\footnote{Except for $d=2$, the bounds of $p_{\mbox{\tiny NL,SLO}}$ presented here were obtained through numerical optimizations. Incidentally, within the numerical precision of these optimizations, these bounds coincide with the respective threshold value obtained using the local filtering protocol given in Ref.~\cite{SP}, i.e., for $d=3$, 4, 5, and 6, we have respectively $\frac{4}{17}(3\sqrt{2}-1)$, $\frac{3}{7}(2\sqrt{2}-1)$, $\frac{8}{41} (5\sqrt{2}-3)$  and $\frac{5}{14} (3\sqrt{2}-2)$.} & $0.7071$ & 0.7630 & 0.7837   & 0.7944 & 0.8009  \\                  \hline	
	\end{tabular}
        \end{ruledtabular}
	\caption{\label{Tbl:CriticalP} Critical weight $p^*$ above which the nonlocal behavior of the $d$-dimensional Werner state $\tau_{\mbox{\tiny W},d}(p)$ can be demonstrated using the protocol described in Figure~\ref{Fig:Protocol} and the help of ancillary $\rho$ that has positive partial transposition. Also shown are $p_{\mbox{\tiny sep}}$, the maximum value of $p$ whereby $\tau_{\mbox{\tiny W},d}(p)$ is separable, $p_{\mbox{\tiny L}}$, the best lower bound  on the maximal  value of $p$ whereby $\tau_{\mbox{\tiny W},d}(p)$ is known to be 1-local (for projective measurements), and $p_{\mbox{\tiny NL,SLO}}$, the smallest value of $p$ whereby $\tau_{\mbox{\tiny W},d}(p)$ is found (numerically) to violate CHSH inequality after local filtering operation. }
\end{table}

\subsection{Multipartite examples}

Let us now provide  a {\em trivial} illustration on how the multipartite activation can be achieved in the case where $\tau$ is {\em not} biseparable with respect to all bipartitions. Imagine that a $d$-dimensional Werner state is shared between parties $\{2,3\}$, i.e.,
\begin{equation}
	\tau=\tau_{\{2,3\}}=\tau_{\mbox{\tiny W},d}(p),
\end{equation}	
and that a $d$-dimensional maximally entangled state $\ket{\Phi_{d}}=\tfrac{1}{\sqrt{d}}\sum_{s=1}^d\ket{s,s}$ is shared between parties $\{1,3\}$. Clearly, parties $\{1,2\}$ cannot violate  CHSH inequality if they only have access to either $\tau_{\{2,3\}}$ or $\ket{\Phi_d}_{\{1,3\}}$. 

Now if party 3 performs a projection onto $\ket{\Phi_d}$  across the two systems that  has access to  $\proj{\Phi_{d}}_{\{1,3\}}\otimes\tau_{\{2,3\}}$, then,  conditioning on a successful projection, the three parties now share the following state:
\begin{equation}
	\kappa=\tau_{\mbox{\tiny W},d}(p)\o\proj{\Phi_d}_{\{3\}},
\end{equation}
where the state $\tau_{\mbox{\tiny W},d}(p)$ is now shared between parties $\{1,2\}$.  Clearly, for $p>p^*$, we can now proceed with the two-party activation protocol described above (see  Sec.~\ref{Sec:Bipartite}) to demonstrate via the parties $\{1,2\}$ the CHSH-nonlocality hidden in $\tau_{\{2,3\}}$.

With some thought, it is clear that if we allow the dimension of the ancillary state $\rho$ to be arbitrarily large, a  protocol similar to that described above can be applied to teleport~\cite{Bennett:Teleportation} any $n$-partite non-fully-separable state to the first two parties. For $\tau$ that is not biseparable with respect to all bipartitions, the resulting state shared between the first two parties is entangled,
and thus one can complete the multipartite activation protocol by proceeding with the two-party activation protocol (see Figure~\ref{Fig:Protocol}). However, one is reminded from Theorem~\ref{Thm} that a multipartite activation of  hidden-CHSH-nonlocality is possible even if the non-fully-separable multipartite state $\tau$ is biseparable with respect to all bipartitions (eg., when $\tau=\sigma_{\mbox{\scriptsize shifts}})$.

Let us also remark that in the trivial multipartite activation protocol described above, one requires an ancillary state $\rho$ that acts on the Hilbert space $[\C^d\otimes\C^d\otimes\C^2]\otimes[\C^d\otimes\C^2]\otimes[\C^d]$, where the state space $[\C^d]\otimes[\C^1]\otimes[\C^d]$ arises from the resource for teleportation and the state space $[\C^d\otimes\C^2]\otimes[\C^d\otimes\C^2]\otimes[\C^1]$ stems from the two-party ancillary state in $\c$. However, in the protocol that we have adopted for the proof of Theorem~\ref{Thm}, it is clear that an ancillary $\rho$ that acts on $[\C^2]\otimes[\C^d\otimes\C^2]\otimes[\C^d]$ is sufficient. As a concrete example, we note that for $d=2$ and any $p>4\sqrt{2}-5$, the ancillary state $\rho_{\{1,2,3\}}$ acting on $[\H_1'']\otimes[\H_2'\otimes\H_2'']\otimes[\H_3']$,
\begin{equation}
	\rho_{\{1,2,3\}} = \frac{1}{16} \sum_{i, j =0}^3 R_{i j} \, \sigma_j \otimes \sigma_i \otimes \sigma_j \otimes \sigma_i
\end{equation} 
with $R$ defined in Eq.~\eqref{eqrho} can be used to demonstrate the hidden-CHSH-nonlocality of both  $\tau_{\{2,3\}}$ and $\rho_{\{1,2,3\}}$ if we set (Figure~\ref{Fig:Scenario}) 
\begin{equation}
	F_1=\id_{\H_1''},\quad F_2= \ket{\Phi_{\H_{2} \H_{2}'}}\otimes \id_{\H_2''},\quad F_3=\ket{\Phi_{\H_{3} \H_{3}'}}.
\end{equation}

\section{Discussion}\label{Sec:Discussion}

Determining whether a given quantum state can exhibit non-locally-causal correlations is a notoriously difficult problem. Although there exist (heuristic) algorithms for determining if a quantum state $\tau$ can violate any given Bell inequality (see, eg. Ref.~\cite{Liang:2007} and references therein), or a large class of Bell inequalities~\cite{Terhal:2003}, such an approach is not guaranteed to determine with certainty whether $\tau$ can indeed exhibit nonlocal behavior. Note that this question is not only of fundamental interest, but also has practical implication for the usefulness of any given quantum state in the reduction of communication complexity~\cite{QCC} or quantum key distribution~\cite{QKD}.

Going beyond the standard scheme, we show in this paper that {\em all} entangled, i.e., not fully-separable multipartite quantum states are capable of exhibiting CHSH-nonlocality when assisted by an ancillary quantum state which by itself cannot be used to demonstrate CHSH nonlocality. In other words, for each non-fully-separable  state $\tau$, there exists a Bell-type scenario in which $\tau$ cannot be substituted by an arbitrarily large amount of classical correlations (or equivalently, a fully-separable state). In this sense, every entangled state (together with the ancillary state involved in the Theorem) is capable of exhibiting correlations that {\em cannot be simulated} when the spatially separate parties only have access to shared randomness.

Naturally, the non-deterministic preprocessing involved in proving our key result reminds one of the detection loophole discussed in a standard Bell test. An important distinction between the two, as was pointed out by Popescu~\cite{SP}, and also by \.{Z}ukowski {\em et al.}~\cite{Zukowski:Loophole}, is that in the demonstration of hidden nonlocality, this non-deterministic element takes place {\em before} the Bell test. Therefore, a priori, the {\em pre-selection} that results from the non-deterministic process {\em does not} causally depend on the choice of measurements made subsequently. On the contrary, in the case of a detection loophole that arises from inefficient detectors, the {\em post-selection} actually takes place {\em after} the choice of measurements is decided. As a result, a standard Bell test is free of the detection loophole only if the overall detection efficiency is above a certain threshold, whereas in the demonstration of hidden nonlocality, this pre-selection efficiency can be arbitrary low (as long as it is non-zero). For a more rigorous discussion of this distinction, we refer the reader to the detailed discussion presented in Ref.~\cite{Zukowski:Loophole}.

Clearly, the ancillary state $\rho$ employed in the Theorem is only guaranteed to not violate the CHSH Bell inequality but not any other Bell inequalities. Nonetheless, even if $\rho$ violates another Bell inequality, it cannot display, by itself, any nonlocal correlations when tested with the CHSH inequality. Therefore, the violation of $\rho\otimes\tau$ for any entangled but CHSH-preprocessed-1-local $\tau$ still manifests the CHSH-nonlocality {\em hidden} in both states. A natural way to strengthen the current result thus consists of considering only ancillary states which by themselves do not violate {\em any} Bell inequality, a possibility that we shall leave as an open problem.

Finally, since our result applies to all (finite-dimensional) multipartite entangled states, it is also natural to ask if there exists an analog to Theorem~\ref{Thm}, and hence examples of superactivation for genuine multipartite nonlocality~\cite{GMNL}. We conjecture that there are examples of this kind but we shall leave this for future research.

\begin{acknowledgements}
We acknowledge useful discussions with Nicolas Gisin. This work was supported by the  Swiss NCCR-QSIT, the CHIST-ERA DIQIP and CatalunyaCaixa.
\end{acknowledgements}

\appendix

\section{Proofs of Lemmas}

\subsection{Proof of Lemma~\ref{Lem:C}}
\label{App:Lemma2}

\begin{proof}[Proof of Lemma~\ref{Lem:C}]
In Ref.~\cite{MLD}, it was shown that, a two-qubit state $\varrho$ violates CHSH if and only if there exists $U, V\in\rm SU(2)$ and $\theta\in [0,\pi/4]$ such that
\begin{equation}\label{Cqubits}
    \tr\!\left[\varrho\, (U \otimes V)^\dag\, H_\theta\,
    (U \otimes V) \right]
    < 0\ .
\end{equation}
Assume that the $n$-party state $\rho$ violates CHSH between parties $\{1,2\}$ after some SLO on all parties. This implies that there is a separable map $\Omega$ such that the two-party state ${\rm tr}_{3,\cdots. n} [\Omega(\rho)]$ violates CHSH, where ${\rm tr}_{3,\cdots, n}$ stands for the trace over parties $\{3, \ldots, n\}$. Without loss of generality, we can assume that the output of $\Omega$ is a qubit for all but the first two parties.

In Ref.~\cite{Masymp} it was proven that, if a bipartite state violates CHSH then it can be transformed by stochastic local operations into a two-qubit state which also violates CHSH. Therefore, there
must exist a bipartite separable map $\Omega\rq{}$ with a two-qubit output such that
the two-qubit state $\Omega\rq{} ({\rm tr}_{3,\cdots, n} [\Omega(\rho)])$ satisfies condition (\ref{Cqubits}) for some $(U, V, \theta)$. Since ${\rm tr}_{3,\cdots, n}$ commutes with $\Omega\rq{}$, there is a separable map $\Omega\rq{}\rq{}$ with qubit output for all the $n$ parties such that $\Omega\rq{} ({\rm tr}_{3,\cdots, n} [\Omega(\rho)]) = {\rm tr}_{3,\cdots, n} [\Omega\rq{}\rq{} (\rho)]$, and then
\begin{equation}
    \tr\!\left[ \Omega\rq{}\rq{} (\rho) \left( [(U \otimes V)^\dag H_\theta (U \otimes V)] \otimes \id_{\C^2}^{n-2} \right) \right]
    < 0\ .
\end{equation}
The above implies that there is a Kraus operator in $\Omega\rq{}\rq{}$, denoted by ${\bf F} =  F^1 \otimes \cdots \otimes F^n$, such that
\begin{equation}
    \tr\!\left[ {\bf F} \, \rho\, {\bf F}^\dag 
    \left( [(U \otimes V)^\dag H_\theta (U \otimes V)] \otimes 
    \id_{\C^2}^{n-2} \right) \right]     < 0\ .
\end{equation}
Note that $U,V$ can be absorbed into the definition of  $F^1, F^2$ giving
\begin{equation}
    \tr\!\left[ {\bf F} \, \rho\, {\bf F}^\dag \left( H_\theta \otimes \id_{\C^2}^{n-2} \right) \right]
    < 0\ .
\end{equation}
This proves one direction of the lemma, the proof for the other direction is trivial.
\end{proof}

\subsection{Lemma~\ref{Lem:SeparableMap} and its proof}
\label{App:Lemma4}

\begin{lemma}\label{Lem:SeparableMap}
Let $\Omega_x$  be a family of completely positive maps, with input $[\C^2]^{\otimes n}$ and output $\bigotimes_{i=1}^n [\H_{i}\otimes \C^2]$, and being separable with respect to the partition denoted by the brackets.  Let $\mu$ be a unit-trace, positive semi-definite matrix
acting on $\bigotimes_{i=1}^n \H_{i}$ such that
\begin{equation}
\label{principal}
    \mu\t\! \otimes H_{\pi/4}\otimes\id_{\C^2}^{n-2} - \int\!\! dx
    \ \Omega_x\!\left( H_{\theta_x}\otimes\id_{\C^2}^{n-2} \right)
    \geq 0\ ,
\end{equation}
where $H_\theta$ is defined in (\ref{H}), then $\mu$ has to be
fully separable.
\end{lemma}

\begin{proof}
First, let us characterize the solutions
$\Omega_x$ of (\ref{principal}). The Bell basis is defined as
\begin{eqnarray}\label{bell}
  \ket{\Phi_{^1_2}} &=& 2^{-1/2}\left( \ket{0,0} \pm \ket{1,1} \right)\ , \\
  \ket{\Phi_{^3_4}} &=& 2^{-1/2}\left( \ket{0,1} \pm \ket{1,0} \right)\ .
\end{eqnarray}
The matrices $H_\theta$ are diagonal in this basis, $H_\theta=
\sum_{r=1}^4 N_\theta^r\, \Pi_r$, where $\Pi_r \equiv\proj{\Phi_r}$
are the Bell projectors and $N_\theta^i$ are the components of the vector
\begin{equation}\label{N}
    N_\theta = \left[
\begin{array}{c}
  1-\cos\theta-\sin\theta \\
  1+\cos\theta-\sin\theta \\
  1-\cos\theta+\sin\theta \\
  1+\cos\theta+\sin\theta \\
\end{array}
\right].
\end{equation}
For each value of $x$ consider the sixteen matrices acting on the space $\bigotimes_{i=1}^n \H_i$ given by
\begin{equation}\label{omega}
    \omega_x^{rs} \equiv  \frac{1}{2^{n-2}} 
    \tr_{\otimes_{i=1}^n \H_i''}\!\!
    \left[\left(\id_{\H}\o\Pi_r\otimes\id_{\C^2}^{n-2}\right)\,
    \Omega_x(\Pi_s\otimes\id_{\C^2}^{n-2})\right]\ ,
\end{equation}
for $r,s=1,\ldots, 4$, where the identity matrix $\id_{\H}$ acts on $\bigotimes_{i=1}^n \H_i$, and the Bell projectors $\Pi_r$ act on
$\H_{1}''\otimes\H_{2}''$. Each $\omega_x^{rs}$ is the result of a
physical operation, and hence positive---although not necessarily normalized. One can see $\omega_x^{rs}$ as
the Choi-Jamio{\l}kowski state corresponding to the map $\Omega_x$, after
``twirling" the input and output subsystems 1 and 2 with respect to the group of
unitaries that leaves Bell-diagonal states invariant. Multiplying the left-hand-side of \eqref{principal} by $\id_\H \otimes \Pi_i \otimes \id_{\C^2}^{n-2}$ and taking the trace over $\otimes_{i=1}^n \H_i''$, we get
\begin{equation}\label{uk}
    \mu\t N_{\pi/4}^r - \int\!\! dx
    \sum_{s=1}^4\ \omega_x^{rs} N_{\theta\!_x}^s \geq 0
    \ ,
\end{equation}
for $r=1,\ldots 4$. Denote by $M_x$
the $4\times4$ matrix with components $M_x^{rs}=\tr[\omega_x^{rs}]$.
Performing the trace on the left-hand-side of (\ref{uk}) we obtain the
four inequalities
\begin{equation}\label{io}
    N_{\pi/4} - \int\!\! dx
    \ M_x \cdot N_{\theta\!_x} \succeq \nv\ ,
\end{equation}
where $\nv$ is the 4-dimensional null vector, and the symbols $\cdot$
and $\succeq$ mean, respectively, standard matrix multiplication and
component-wise inequality. 

In what follows we consider the set of $4\times 4$ matrices $M$ obtained by taking any separable map $\Omega$ and projecting it as
\begin{equation}
	M^{rs} =  \frac{1}{2^{n-2}} \tr
    \left[\left(\id_{\H}\o\Pi_r\otimes\id_{\C^2}^{n-2}\right)
    \Omega (\Pi_s\otimes\id_{\C^2}^{n-2})\right]\ .
\end{equation}
This defines a linear transformation mapping any $\Omega$ to an $M$ matrix. Since the set of separable maps~(\ref{sepmap}) is a convex cone, the set of matrices $M$, denoted ${\cal M}$ is a convex cone too. Any separable map is a positive linear combination of maps with a single Kraus operator $\Omega(\rho)= {\bf F}\, \rho\, {\bf F}^\dag$, where ${\bf F}= \left( \bigotimes_{i=1}^n F_{i} \right)$. In this case
\begin{align*}
	M^{rs} &=  \frac{1}{2^{n-2}}\, \tr\!
    \left[\left(\id_{\H}\o\Pi_r\otimes\id_{\C^2}^{n-2}\right)
	{\bf F} (\Pi_s\otimes\id_{\C^2}^{n-2})
    {\bf F}^\dag \right]
\\	&=  \nu\, \tr\!
    \left[\left(\id_{\H_1} \o \id_{\H_2} \o\Pi_r\right)\,
	(F_1 \otimes F_2) \Pi_s (F_1 \otimes F_2)^\dag \right]
\end{align*}
where
\begin{equation*}
	\nu =  \frac{1}{2^{n-2}} \prod_{i=3}^n \tr [ F_i F_i^\dag ]\ .
\end{equation*}
The above equation shows that the convex cone ${\cal M}$ for any $n \geq 2$ is identical to the one for $n=2$. The characterization of ${\cal M}$ for the $n=2$ case was obtained in Ref.~\cite{DLM}, and goes as follows. Denote by $\mathcal{D}$ the set of $4\times 4$ doubly-stochastic matrices, that
is, the convex hull of the permutation matrices
\cite{Bhat}. Denote by $\mathcal{G}$ the convex hull of all matrices
obtained when independently permuting the rows and columns of
\begin{equation}\label{e2}
G_0\equiv \left[
\begin{array}{cccc}
  1 & 1 & 0 & 0 \\
  1 & 1 & 0 & 0 \\
  0 & 0 & 0 & 0 \\
  0 & 0 & 0 & 0 \\
\end{array}
\right]\ .
\end{equation}
It was shown in Ref.~\cite{DLM} that, any matrix $M$ as defined above, in the $n=2$  case, can be
written as
\begin{equation}\label{m}
	{\cal M}= \{pD+qG\, |\, p,q \geq 0, D\in \mathcal{D}, G\in \mathcal{G} \}\ .
\end{equation}
Now we know that equation~(\ref{io}) is independent of $n$. It is shown in Ref.~\cite{MLD} that any solution of~(\ref{io}) satisfies $\theta_x = \pi/4$ for all $x$, and $\int dx\, M_x =M_0$ where
\begin{equation}\label{M_0} M_0 = \left[
\begin{array}{cccc}
  1 & 0 & 0 & 0 \\
  0 & 1- \eta & \eta & 0 \\
  0 & \eta & 1- \eta & 0 \\
  0 & 0 & 0 & 1 \\
\end{array}
\right]\ ,
\end{equation}
for some $\eta\in [0,1]$. Note that $M_0 \cdot N_{\pi/4} =N_{\pi/4}$, hence
the left-hand-side of~(\ref{io}) is zero. This implies that the left-hand-side of (\ref{uk}) is
traceless for all $i$ and therefore
\begin{equation}\label{fin}
    \mu\t N_{\pi/4}^r =
    \sum_{s=1}^4\, \omega_0^{rs} N_{\pi/4}^s,\quad r=1,\ldots, 4,
\end{equation}
where $\omega_0$ is any $\omega_x$ that gives rise to $M_0$. Using the
same argument, the pairs $(r,s)$ for which $M_0^{rs}=0$ are such that
$\omega_0^{rs}=0$. Therefore, by adding the equalities in (\ref{fin})
corresponding to $r=2,3$, and using the definition of $\omega_0^{rs}$
in \eqref{omega}, we obtain
\begin{equation}\label{fin2}
    2\, \mu\t =  \frac{1}{2^{n-2}}\tr_{\otimes_{i=1}^n \H_i''}\!\left[\left(\id_{\H}\o\Psi\otimes\id_{\C^2}^{n-2}\right)
    \Omega_0(\Psi\otimes\id_{\C^2}^{n-2})\right]\ ,
\end{equation}
where $\Psi=\Pi_2+ \Pi_3$, and $\Omega_0$ is any $\Omega_x$ that
gives rise to $\omega_0$. Using the Peres-Horodecki separability criterion \cite{Peres:PPT, Hppt} one
can check that the (unnormalized) two-qubit state $\Psi$ is a
separable state. Equation (\ref{fin2}) implies that $\mu\t$ is the
output of a separable map applied to a {\em fully-separable} input state, and
hence is a {\em fully-separable state} as we wanted to prove. 
\end{proof}

\section{Formulating the search of an activating ancillary state $\rho\in\c$ via semidefinite programs}
\label{App:SDP}

For any given entangled state $\tau$, we describe below a semidefinite program that can be used to construct, whenever possible, an entangled state $\rho$ that has PPT with respect to party 1 such that $\rho\,\otimes\,\tau\not\in\c$. We start by noting that with the local filtering operations specified in Eq.~\eqref{Eq:F_i}, the analogous expression for the left-hand-side of Eq.~\eqref{C} for $\theta=\pi/4$ becomes the left-hand-side of Eq.~\eqref{Eq:cond}. If there exists a PPT state $\rho$ such that the left-hand-side of Eq.~\eqref{Eq:cond} is less than zero, then we would have identified a PPT state which exhibits the nonlocal behavior of $\rho\otimes\tau$ via the CHSH inequality. 

Specifically, note that if the optimum value of the following optimization
\begin{eqnarray}
&{\rm minimize}_{\{\rho\}}&\tr\!\left[\rho\, (\tau\t\otimes H_{\pi/4}\otimes\id_{\C^2}^{n-2})\right] \nonumber\\
&{\rm subject \ to} \quad & \quad\rho\ge 0,\quad \rho^{\mbox{\tiny T$_1$}}\ge 0. \label{Eq:SDP}
\end{eqnarray}
is negative (where $\rho^{\mbox{\tiny T$_{1}$}}$ is the partial transposition of $\rho$ with respect to $\H_1$), then $\rho$ is guaranteed to be a PPT state such that $\rho\,\otimes\,\tau$ violates the CHSH inequality. Moreover, if $\tau\in\c$, then a negative value for the above optimization problem also implies that $\rho$ must be entangled (since $\rho\,\otimes\,\tau$ is necessarily in $\c$ if $\rho$ is separable and $\tau\in\c$). Finally, note that the optimization problem~\eqref{Eq:SDP} is a semidefinite program as it involves an optimization over positive semidefinite matrices ($\rho$, in this case) which are only subjected to linear matrix inequality constraints~\cite{SDP}.

\end{document}